%% file: root.tex
\documentclass[letterpaper, 10 pt, conference]{ieeeconf}  

\IEEEoverridecommandlockouts                              

\usepackage{hyperref}
\usepackage{graphics} 
\usepackage{epsfig} 
\usepackage{times} 
\usepackage{amsmath} 
\usepackage{amssymb}  
\usepackage{color}
\usepackage{algorithm}
\usepackage[noend]{algpseudocode}
\usepackage{booktabs,lipsum}
\usepackage{accents}
\usepackage{cite}
\usepackage{amsthm}
\usepackage[dvipsnames]{xcolor}
\usepackage[font=footnotesize,skip=5pt]{caption}

\renewcommand{\baselinestretch}{.99}

\def\aignore{\ateam_\text{rmv}}
\def\djset{\mf{G}}
\def\dteam{\mf{D}}
\def\dboundary{\mf{D}_\text{pair}}
\def\ateam{\mf{A}}
\def\subteam{\mc{S}}   
\def\awinrc{\mc{R}_\text{LG}} 
\def\awinr{\mc{R}_A} 
\def\awinri{\mc{R}_I} 
\def\awinrc{\mc{R}_C} 
\def\dwinr{\mc{R}_D} 

\def\target{\mc{T}}

\def\dwinrp{\mc R_\text{pair}}

\def\dspeed{\bar{v}_D}

\def\Qest{Q_\text{LG}}
\def\Qmm{Q_\text{MM}}

\def\Qint{Q_\text{low}}
\def\control{\mf{v}}
\def\controld{\omega}
\def\controlset{\bs{\Gamma}}

\def\curve{\bs{\gamma}}

\def\lgsi{\hat{q}}
\def\curve{\bs{\gamma}}

\def\xa{\mf x_A}
\def\di{D_i}
\def\dj{D_j}
\def\sdi{s_{\di}}

\def\implicit{implicit assignment}
\def\sequential{sequential paired defense}

\newcommand{\T}{^{\mbox{\tiny\sf{T}}}}

\theoremstyle{plain}
\newtheorem{corollary}{Corollary}
\newtheorem{theorem}{Theorem}

\newtheorem{lemma}{Lemma}
\newtheorem{definition}{Definition}
\newtheorem{property}{Property}

\newtheorem*{problem}{Problem}

\def\bs{\boldsymbol}
\def\mf{\mathbf}

\def\mc{\mathcal}
\def\beq{\begin{equation*}}
\def\eeq{\end{equation*}}
\def\bql{\begin{equation}}
\def\eql{\end{equation}}
\def\bqn{\begin{eqnarray*}}
\def\eqn{\end{eqnarray*}}
\def\bnl{\begin{eqnarray}}
\def\enl{\end{eqnarray}}
\def\bna{\bql\begin{array}{rcl}}
\def\ena{\end{array}\eql}
\def\bnn{\beq\begin{array}{rcl}}
\def\enn{\end{array}\eeq}
\def\bma{\begin{bmatrix}}
\def\ema{\end{bmatrix}}
\def\bmx{\begin{matrix}}
\def\emx{\end{matrix}}
\def\ben{\begin{enumerate}}
\def\een{\end{enumerate}}
\def\bit{\begin{itemize}}
\def\eit{\end{itemize}}
\def\bei{\begin{itemize}}
\def\eei{\end{itemize}}
\def\bet{\begin{tabular}}
\def\eet{\end{tabular}}

\newcommand{\allcaps}[1]{\uppercase\expandafter{#1}}

\def\bfs{\begin{footnotesize}}
\def\efs{\end{footnotesize}}
\def\bss{\begin{small}}
\def\ess{\end{small}}

\DeclareMathOperator*{\argmax}{arg\,max}

\def\Appendix{Appendix}
\def\RALversion{0}

\title{\LARGE \bf
Cooperative Team Strategies for Multi-player Perimeter-Defense~Games
}

\author{Daigo Shishika$^{*}$, James Paulos and Vijay Kumar
\thanks{
We gratefully acknowledge the support of ARL grant ARL DCIST CRA W911NF-17-2-0181.
}
\thanks{The authors are with GRASP Laboratory at the University of Pennsylvania {\tt\small \{shishika,jpaulos,kumar\}@seas.upenn.edu}}
}

\begin{document}

\maketitle
\thispagestyle{empty}
\pagestyle{empty}

\begin{abstract}
This paper studies a variant of the multi-player reach-avoid game played between intruders and defenders with applications to perimeter defense. The intruder team tries to score by sending as many intruders as possible to the target area, while the defender team tries to minimize this score by intercepting them. Finding the optimal strategies of the game is challenging due to the high dimensionality of the joint state space, and the existing works have proposed approximation methods to reduce the design of the defense strategy into assignment problems. However they suffer from either suboptimal defender performance or computational complexity. Based on a novel decomposition method, this paper proposes a scalable (polynomial-time) assignment algorithm that accommodates cooperative behaviors and outperforms the existing defense strategies. For a certain class of initial configurations, we derive the exact score by showing that the lower bound provided by the intruder team matches the upper bound provided by the defender team, which also proves the optimality of the team strategies.
\end{abstract}

\section{Introduction}

\noindent
Surveillance of perimeters and securing perimeters are important tasks in civilian and military defense applications.
It has become practical to deploy large number of autonomous agents to address these problems using multi-robot systems.
Approaches to counter intrusions by unmanned vehicles have been studied including detection and tracking mechanisms~\cite{Ganti2016},
patrolling (scheduling) scheme to visit points on the perimeter~\cite{Agmon2008icra},
and GPS spoofing to manipulate the behavior of the agents~\cite{Kerns2014}.

Scenarios with evasive targets who need to be detected, intercepted, or surrounded by other robots are often formulated as \emph{pursuit-evasion games}.
The two-player version of these games (one evader vs.\ one pursuer) may be formulated as a differential game and solved using the Hamilton-Jacobi-Isaacs equation \cite{Isaacs}.
However, the high dimensional state of multi-player games can make this approach intractable.
Multiple practical approaches to analyzing the multi-player game have been considered
including Voronoi tessellation \cite{Bakolas2010,Pierson2017}, 
cyclic pursuit \cite{Bopardikar2009}, and 
vehicle-routing formulation~\cite{Agharkar2014}.

When a pursuer robot is tasked to defend a certain area/target against the evader, the pursuit-evasion game becomes the \emph{target-guarding problem}, first introduced by Isaacs~\cite{Isaacs}.
In these games the evader/intruder tries to reach the target without being intercepted, and the pursuer/defender's goal is to either intercept the intruder or delay its intrusion indefinitely \cite{Rusnak2005,Garcia2015,Liang2019}.
In this paper we study the multi-player version of the target-guarding problem, also called \emph{reach-avoid games} \cite{Huang2011icra, Fisac2015a,Chen2017}.

In related work, the multi-player game has been approximated as the combination of two-player games \cite{Chen2017}.
All pairwise two-player games between agents can be solved using a differential-game formulation using numerical tools, even in the case of complex environment geometry.
The authors then propose a team defense strategy that assigns to each defender one feasible intruder to capture.
This assignment provides an upper-bound on the intruder's score, but the obtained defense policy is potentially suboptimal because it can not consider cooperative maneuvers for capture.

Our prior work extended the assignment strategy by incorporating a cooperative behavior absent from two-player games: a pair of defenders employ a ``pincer movement'' to pursue an intruder from both sides \cite{shishika2018cdc}.
Although the defense performance was improved, the policy required a solution to the maximum-independent-set problem which is NP-hard.

This paper extends the decomposition method first introduced in \cite{shishika2018cdc}, and proposes a polynomial-time algorithm that outperforms the existing defense policies.
{We also discuss the optimality of the strategies by deriving a condition under which the lower-bound provided by the intruder team matches with the upper-bound provided by the defender team.}
To the best of our knowledge, we are the first to show such optimality (equilibrium strategies) in a pursuit-evasion game played between teams.


The two main contributions of this paper are
(i) the polynomial-time defense strategy that outperforms the existing ones; and
(ii) the analysis that shows the optimality of the proposed strategy, {when the defenders have a numerical advantage and the initial configuration satisfies a certain condition}.
Our video (available at \url{https://youtu.be/6zUPkzh_iPU})
illustrates the complexity of the problem and the effectiveness of the strategy through multiple multi-robot simulations.

Section~\ref{sec:problem} formulates the problem. Section~\ref{sec:background} reviews existing results that we use to build our method.
Section~\ref{sec:decomposition} presents the decomposition method.
Section~\ref{sec:defense_policy} proposes the defense policy and analyzes its performance.
Section~\ref{sec:simulation} provides numerical results.

\section{Problem Statement  \label{sec:problem}}
\noindent
This section formulates the reach-avoid game for robots defending a perimeter.
The target $\mc{T}\subset \mathbb{R}^2$ is assumed to be a convex region on a plane, and its perimeter is given by an arc-length parameterized curve $\curve: [0,L)\rightarrow \partial \mc{T}$, where $L$ denotes the perimeter length.
We use $s\in[0,L)$ to denote the arc-length position on the curve.


A set of $N_D$ robot defenders $\dteam = \{D_i\}_{i = 1}^{N_D}$ are constrained to move on the perimeter.
We assume that the defender robots must stay on the perimeter of the enclosed compound.
The position of the $i$th defender is described by $s_{D_i}$ or $\mf x_{D_i} = \curve(s_{D_i})$.
The defender robot's control input is the signed speed; $\dot{s}_{D_i} = \omega_{D_i}$ with the constraint $|\controld_{D_i}|\leq \dspeed$.

A set of $N_A$ intruders $\ateam = \{A_j\}_{j = 1}^{N_A}$ have first-order dynamics in $\mathbb{R}^2$.
The control inputs are the velocities; $\dot{\mf{x}}_{A_i} = \control_{A_i}$ with the constraint $\| \control_{A_i} \| \leq \bar v_A$.
The ratio of the maximum speeds is described by
\bql
\nu = \frac{\bar v_A}{\bar v_D}  \in(0,1],
\eql
i.e., the defenders are at least as fast as the intruders.
Also, this paper focuses on the case $N_A<N_D$ for conciseness.


An intruder is captured by a defender if their distance becomes zero.
In a microscopic view, $A_i$ scores if it reaches the target ($\mf x_{A_i}\in \partial \target$) without being captured by the defenders.
In this work, we assume that the defender is consumed/eliminated when it captures one intruder.
This is a reasonable assumption because \cite{shishika2018cdc} shows a method for the intruders to ensure one defender can not make multiple captures. 

We assume perfect state-feedback information structure: i.e., the state vector $\mf z = [s_{D_1}.., s_{D_{N_D}},\mf x_{A_1}\T.., \mf x_{A_{N_A}}\T]\T$ is known to every player.
The team strategies $\controlset_D$ and $\controlset_A$ are mappings from the current states $\mf z$ to the control actions $\bs\omega_D = [\omega_{D_1}..,\omega_{D_{N_D}}]$ and $\control_A= [\control_{A_1}.. ,\control_{A_{N_A}}]$ respectively.

\begin{problem}
Given an initial configuration of the game, $\mf z_0$,
what are the optimal team strategies $\controlset_A^*$ and $\controlset_D^*$ for the differential game
\vspace{-5pt}
\bql
\min_{\controlset_D}\max_{\controlset_A} Q(\mf z_0;\controlset_D,\controlset_A),
\eql
and what is the outcome, $Q^*(\mf z_0)=Q(\mf z_0;\controlset_D^*,\controlset_A^*)$?
\end{problem}


\section{Background}\label{sec:background}
\noindent
We examine two limiting cases from our prior work: single defense against one intruder, and pair defense against one intruder \cite{shishika2018cdc,shishika2019shape}.
We also review existing assignment-based defense strategy for the multiplayer game \cite{Chen2017,shishika2018cdc}.


\subsection{One v.s.~One Game}
\noindent 
Given the defender location $s_D$, the game space can be divided into intruder-winning region $\awinr(s_D)$ and the defender-winning region $\dwinr(s_D)$, and the surface that separates the two is called the \emph{barrier} \cite{shishika2019shape}.

If the game starts in a configuration $\xa \in \awinr(s_D)$, there is a point on the perimeter that the intruder can reach before the defender does.
In the perimeter defense game, the barrier for the one vs.\ one (1v1) game is a closed curve that starts and ends at the defender position.
\begin{figure}[t]
\centering
\includegraphics[width=.48\textwidth]
{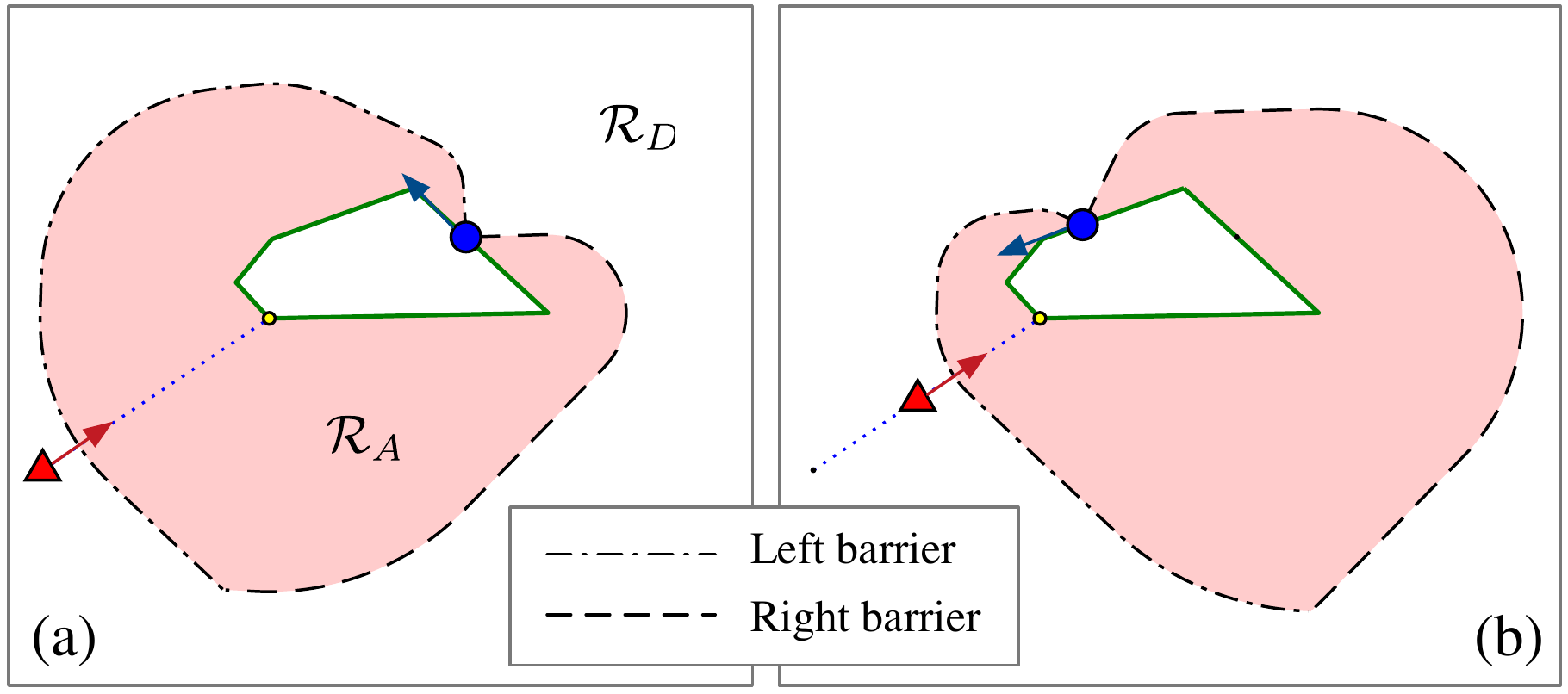}
\caption{
Barriers in 1v1 game for a polygonal perimeter with $\nu=0.7$.
The intruder winning region $\awinr(s_D)$ is colored in red.
(a) Initial configuration with $\xa\in\dwinr$.
(b) Both players moved optimally and the intruder stays in the defender-winning region.
}
\label{fig:1v1}
\vspace{-5pt}
\end{figure}
\begin{property}
\label{prop:1v1barrier}
The barrier consists of left and right barriers.
Starting from $\xa \in \dwinr$, the intruder cannot penetrate the left (resp.~right) barrier so long as the defender is moving counter-clockwise (ccw) (resp.~clockwise (cw)) \cite{shishika2018cdc,shishika2019shape}.
\end{property}
If the initial configuration is $\xa \in \dwinr(s_D)$, the defender can move to keep the intruder inside of $\dwinr(s_D)$.
This implies that the intruder cannot enter $\awinr(s_D)$ -- therefore cannot reach the perimeter -- without being captured.
The construction of the barrier and the feedback strategies are detailed in \cite{shishika2018cdc,shishika2019shape}, but only their existence is sufficient to complete the analysis in this paper.


\subsection{Two v.s.~One Game  \label{sec:two_vs_one}}
\noindent
When we have two robot defenders, let $\di$ and $\dj$ be the defender on the cw and the ccw side.
For conciseness, we denote $\awinr(i)$ to mean $\awinr(\sdi)$.
A naive extension of the 1v1 game will conclude that the intruder-winning region against a pair of defenders is
\bql
\awinri(i,j) = \awinr(i) \cap \awinr(j).
\eql
However, the intrusion strategy and the winning region are now different because the intruder has to avoid both $\di$ and $\dj$ simultaneously.

The actual intruder-winning region in the two vs.\ one (2v1) game, which we denote by $\awinrc(i,j)$, is smaller than $\awinri(i,j)$.
We use the subscript $_I$ for \emph{independent} and $_C$ is for \emph{cooperative}. 
The boundary of $\awinrc(i,j)$ consists of a combination of the left barrier of $\di$, a circle with the center at the midpoint of the two defenders, and the right barrier of $\dj$ (see Fig.~\ref{fig:2v1barrier}).

\begin{figure}[t]
\centering
\includegraphics[width=.48\textwidth]
{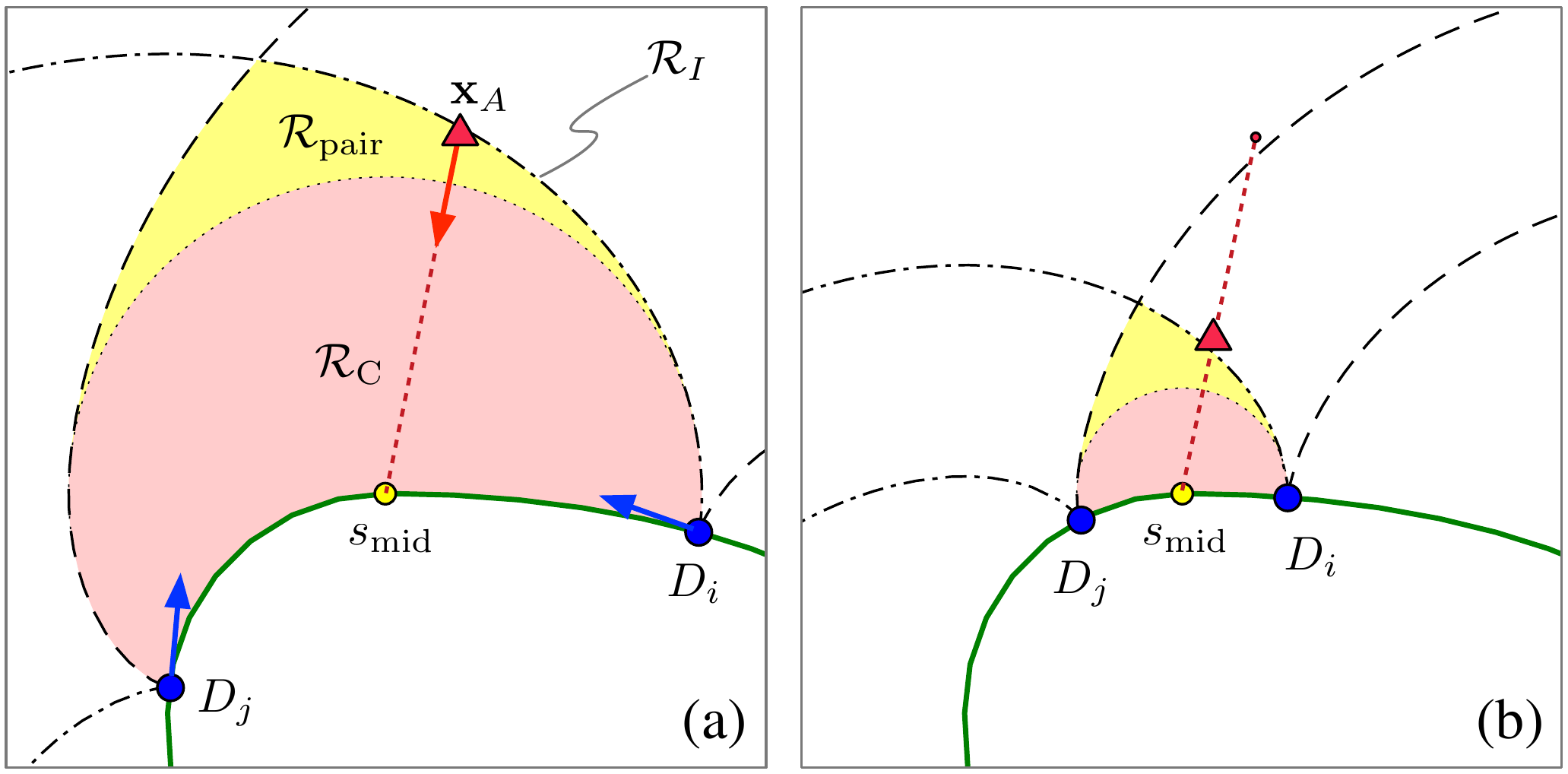}
\caption{
The intruder-winning region $\awinrc$ in the 2v1 game.
(a) The intruder starts in $\dwinrp$ and aims for the midpoint, while the defender robots perform the pincer maneuver.
(b) The intruder eventually enters $\dwinr(\di)$.
}
\label{fig:2v1barrier}
\vspace{-10pt}
\end{figure}
%


The difference with $\awinri(i,j)$ is given by the \emph{paired-defense region} defined as:
\bql
\dwinrp(i,j) \triangleq \awinri(i,j) - \awinrc(i,j). \label{eq:cooperative_defense_region}
\eql
\begin{property}\label{prop:pincer}
If $\xa\in \dwinrp(i,j)$, and if the defenders use a pincer maneuver $[\omega_{\di},\omega_{\dj}]=[1,-1]$, then $\mf x_A \in \dwinr (i)$ or $\mf x_A \in \dwinr (j)$ occurs before the intruder reaches the perimeter
\cite{shishika2018cdc,shishika2019shape}.
\end{property}
If the game starts in the configuration $\xa \in \awinrc$, then there exists a breaching point (the midpoint between the two defenders) that the intruder can reach before the two defenders. 


\subsection{Existing Assignment Strategies}
\noindent 
For a given initial configuration $\{ \mf x_{A_i} \}_{i=1}^{N_A}$ and $\{ \mf x_{D_j} \}_{j=1}^{N_D}$, the defender-winning regions can be used to determine a set of intruders that each defender can potentially win against.

\paragraph*{Maximum Matching \cite{Chen2017}}
Consider a bipartite graph with $\dteam$ and $\ateam$ as two sets of nodes.
We draw an edge between $\di$ and $A_j$ if capture can be guaranteed: $\mf x_{A_j}\in\dwinr(\di)$.
{Maximum-cardinality Matching (MM) refers to finding a largest set of edges such that there is at most one edge extending from each node.}
By performing MM on this graph, we assign at most one intruder to each defender.

The cardinality of the matching, $N_\text{MM}^{cap}$, tells us that at least $N_\text{MM}^{cap}$ intruders will be captured.
The upper bound on the intruder score is then given by 
\begin{gather*}
Q_\text{MM}= N_A - N_\text{MM}^{cap}.
\end{gather*}
The MM assignment can be found in polynomial time \cite{Chen2017}.
However, this method assumes that all defenders play independent games and ignores the possibility of cooperative 2v1 defenses.

\paragraph*{Maximum Independent Set}
Our prior work extended the above assignment method by incorporating the 2v1 games \cite{shishika2018cdc}.
The bipartite graph is augmented with additional nodes representing pairs of defenders.
To avoid conflicting assignments where a single defender is used in both 1v1 and 2v1 defense, an assignment method based on Maximum-Independent-Set (MIS) formulation was proposed.
The upper bound from this method is denoted by $Q_\text{MIS}$, and we have $Q_\text{MIS} \leq Q_\text{MM}$.
The downside of the MIS assignment strategy is its computational complexity (NP-hard), which motivates the derivation of a scalable defender team algorithm.

This paper proposes a polynomial-time algorithm with a score bound $\Qest$ that satisfies $\Qest \leq Q_\text{MIS}\leq Q_\text{MM}$.
In addition, while the above methods only provide upper bounds on the score, this paper discusses the optimality of the team strategy.
This is done by also considering the lower bound on the score guaranteed by the intruder team.


\section{Local Game Decomposition \label{sec:decomposition}}
\noindent 
We describe a method to partition the game space into subregions where ``local games'' between $n_D$ defenders and $n_A$ intruders are played.
This decomposition leads directly to an intruder team strategy and accompanying lower bound on the intruder's score.

\subsection{Local Game Region}
\noindent
We call the intruder winning region $\awinrc(i,j)$ in the 2v1 game to be the \emph{local game region} (LGR), and $\awinri(i,j)$ to be the Independent-LGR (I-LGR).
\begin{figure}[t]
\centering
\includegraphics[width=.48\textwidth]
{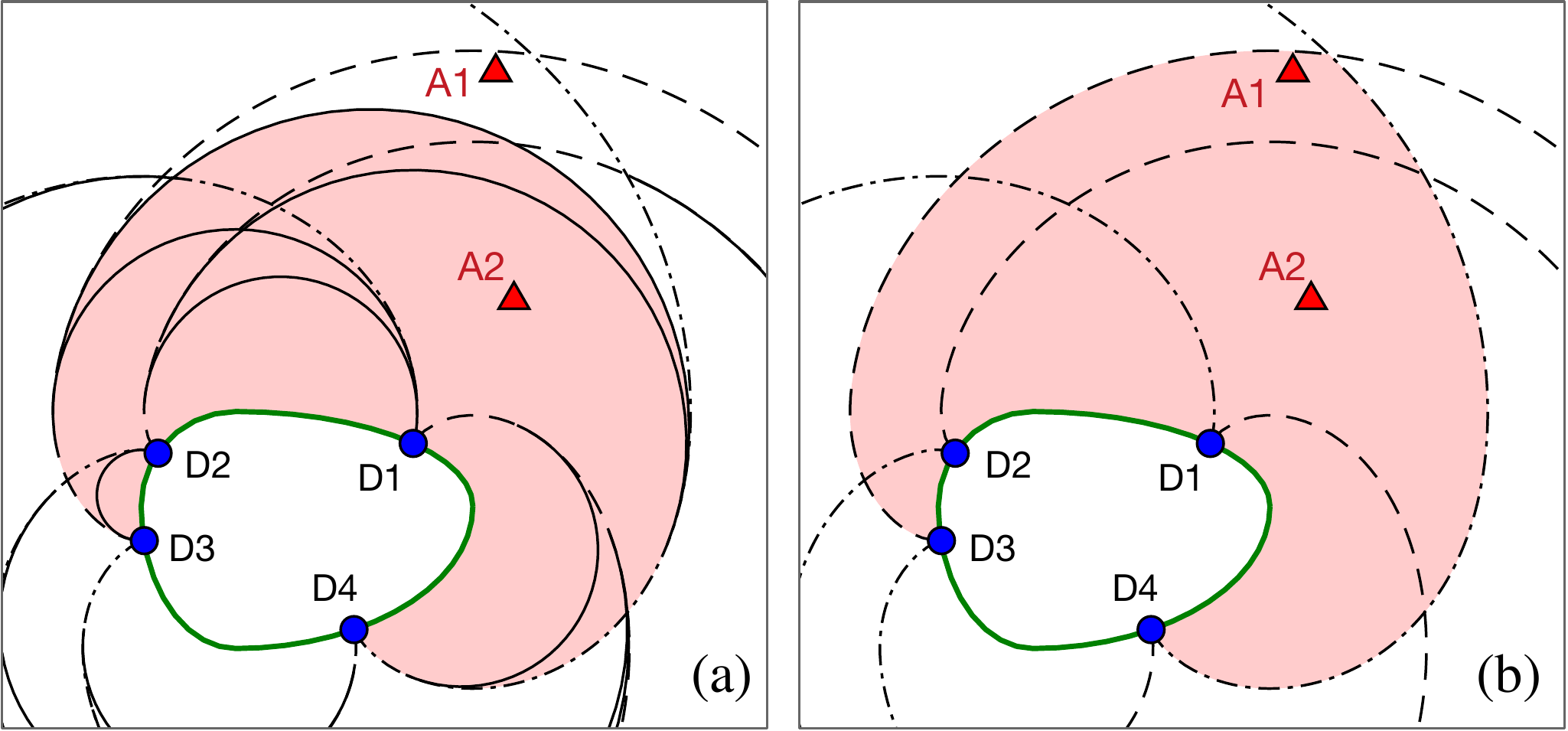}
\caption{
(a) Boundaries of LGRs.
The $7$th region $\awinrc^7$ with $\dboundary^7=(D_4,D_3)$ is highlighted.
The $7$th subteams are $\subteam_D^7=\{D_1,D_2\}$ and $\subteam_A^7=\{A_2\}$.
(b) Boundaries of I-LGRs.
The $7$th intruder subteam is  $\hat{\subteam}_A^7=\{A_1, A_2\}$.
}
\label{fig:lgrs}
\vspace{-5pt}
\end{figure}
These regions are generated by an ordered pair $\dboundary=(\di,\dj)$, and we refer to $D_i$ and $D_j$ as the \emph{right-} and \emph{left-boundary defender}.
In the degenerate case $i=j$ we set $\awinri(i,i)=\awinrc(i,i)=\awinr(i)$.
The total number of LGRs (or I-LGRs) is $N_D^2$, and we use the superscripts $k\in\{1,...,N_D^2\}$ to enumerate them.
We also use $\dboundary^k = (D_R^k, D_L^k)$ to denote the right- and left-boundary defenders of $k$th LGR (or I-LGR).
Note the following properties that are straightforward to see:
\begin{property}
The boundary of an I-LGR consists of the left barrier of $D_R^k$ and the right barrier of $D_L^k$.
\end{property}
\begin{property}
The intersection of two I-LGRs is an I-LGR.
\end{property}

We define the $k$th defender \emph{sub-team}, $\subteam_D^k$, to be the subset of defenders that are on the ccw segment from $s_{D_R^k}$ to $s_{D_L^k}$  (not including the boundary defenders).
The intruder subteam, $\subteam_A^k$ (resp.~$\hat{\subteam}_A^k$), is the set of intruders contained in $\awinrc^k$ (resp.~$\awinri^k$).
See Fig.~\ref{fig:lgrs} for an example.

We use $n_D^k\triangleq |\subteam_D^k|$, $n_A^k\triangleq |\subteam_A^k|$, and $\hat{n}_A^k\triangleq |\hat{\subteam}_A^k|$ to denote the cardinality of the subteams.
Since $\awinrc^k\subseteq\awinri^k$, we have the relation $n_A^k \leq \hat{n}_A^k$.
Also note that the difference $\Delta n_A^k \triangleq \hat{n}_A^k- n_A^k$ comes from the intruders inside $\dwinrp^k$.
We use the LGRs to construct the intruder team strategy and the lower bound on $Q$, whereas the I-LGRs will later be used in the defender team strategy.

\subsection{Local Intrusion Strategy \label{sec:local_intrusion_strategy}  }
\noindent 
By playing a 2v1 game against the boundary defenders,
any intruder in $\awinrc^k$ has a strategy  to win against all defenders except for those in $\subteam_D^k$.  In other words, the intruders can play a local game that involves only $\subteam_D^k$ and no other defenders.
\begin{figure}[t]
\centering
\includegraphics[width=.48\textwidth]
{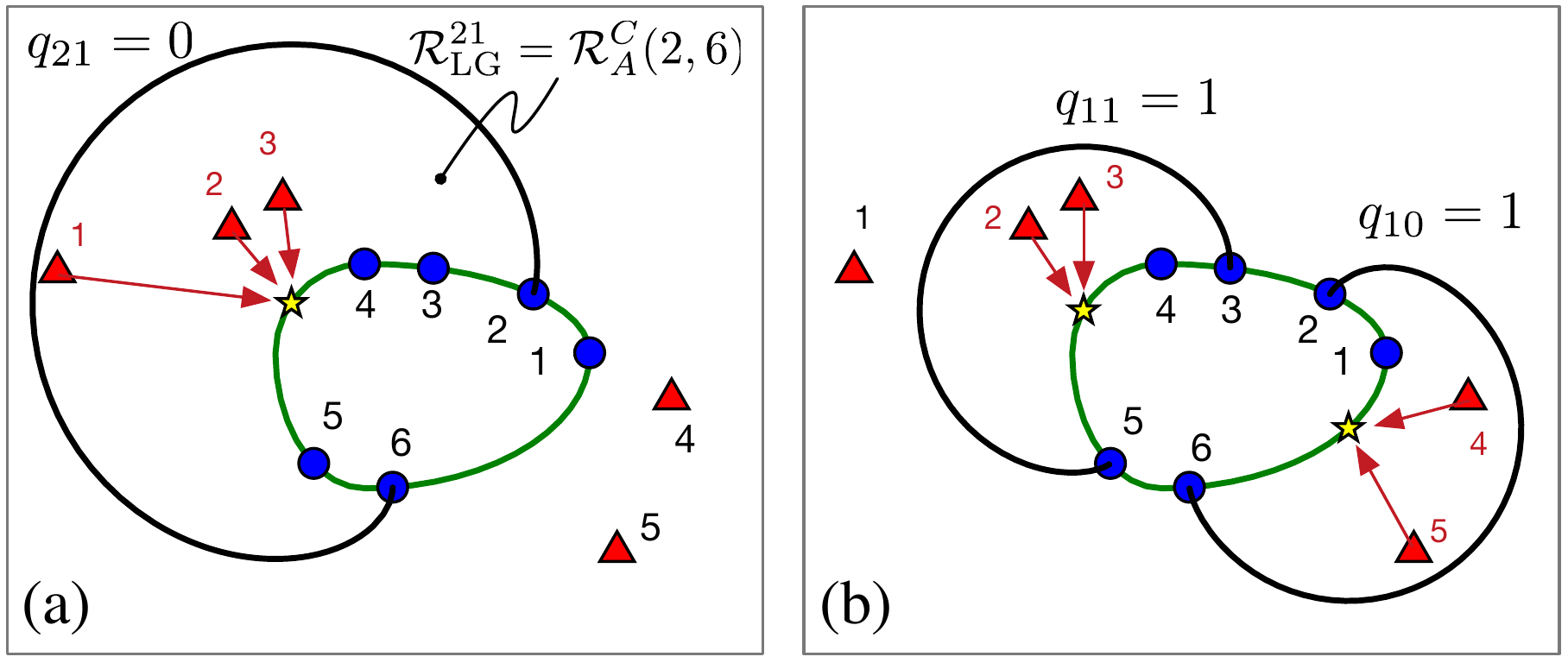}
\caption{
(a) Local intrusion strategy (arrows) for the subteam corresponding to the $21$st LGR.
The score is $q_{21}=3-3=0$.
(b) Optimal partitioning $\djset^*=\{10,11\}$ that guarantees the score $\Qest=q_{10}+q_{11}=2$.
}
\label{fig:subteams}
\vspace{-10pt}
\end{figure}
%

\begin{lemma}[{Local game score}] \label{lem:lgs}
Define $q_k \triangleq n_A ^k- n_D^k$ to be the local game score (LGS).
If $\exists \;k$ such that $q_k>0$,  then the intruder team can score at least $q_k$ points.
\end{lemma}
\begin{proof}
If the $n_A^k$ intruders in $\awinrc^k$ all play the 2v1 game against (approach near the midpoint of) the boundary defenders 
$\dboundary^k$,
then the defenders outside of this LGR cannot reach those intruders.
This implies that at most $n_D^k$ intruders are captured, since we assume that each defender can capture at most one.
Hence, the remaining $q_k$ intruders score.
\end{proof}
Extending this idea to the entire team, the intruders can partition $\ateam$ into disjoint sub-teams and play separate local games.
The optimal partitioning is discussed next.

\subsection{Guaranteed Total Score }\label{sec:QLG}
\noindent
We consider a partitioning of the intruder team based on LGRs.
This approach is certainly not the only way to partition the team, but we later show in Sect.~\ref{sec:saddle_point} that it results in an optimal intruder performance.
Let $\djset \subset \{1,...,N_D^2\}$ denote a  set of disjoint LGRs: i.e., $\awinrc^k\cap\awinrc^l=\emptyset$ for $k,l \in \djset$.
The set $\djset$ gives us a way to partition the intruders into $|\djset|$ disjoint subteams.

Combining $\djset$ with the local intrusion strategy guarantees a lower bound $Q_\text{low}$ on the overall score:
\vspace{-0.03in}
\bql
\Qint(\djset) \triangleq \sum \limits_{k\in \djset} \max\{q_k,0\}.
\vspace{-0.03in}
\eql
An optimal $\djset^*$ is given by $\djset^*\triangleq \argmax_{\djset} \Qint (\djset)$, and it gives the \emph{guaranteed total score}:
\vspace{-0.03in}
\bql
\Qest \triangleq \Qint(\djset^*) =  \max \limits_{\djset} \left( \sum_{k\in \djset} \max\{q_k,0\} \right).
\label{eq:Qest}
\vspace{-0.03in}
\eql

Values for $\Qest$ and $\djset^*$ can be obtained in $O(N_D^4)$ time by recognizing \eqref{eq:Qest} as an instance of the maximum weight independent set problem on a circular arc graph \cite{shishika2018cdc}.
For applications where it is critical to avoid any intrusion, it is easy to test whether the intruders can guarantee a score of at least one: \(\Qest > 0 \Leftrightarrow \exists~q_k > 0\).

We let $\controlset_A^*$ denote the intruder strategy corresponding to the teaming into $\djset^*$ and then approaching the midpoint between the boundary defenders of the associated LGR.
\begin{theorem} \label{thm:Qlow}
The intruder team strategy $\controlset_A^*$ guarantees that the intruders score at least $\Qest$ defined in \eqref{eq:Qest}, i.e.,
\vspace{-0.03in}
\bql
Q(\mf z_0; \controlset_D,\controlset_A^*)\geq  \Qest(\mf z_0), \; \forall \, \controlset_D\in U_D.\label{eq:uAguarantee}
\vspace{-0.03in}
\eql
where $U_D$ is the set of all permissible defender strategies.
\end{theorem}
\begin{proof}
Lemma~\ref{lem:lgs} and the discussion above.
\end{proof}

Now the question is whether the defenders can prevent the intruders from scoring any more than $\Qest$ -- is it also an upper bound?
We address this question in Sec.~\ref{sec:defense_policy}.


\input{extended_structures.tex}

\input{defense_algorithm2.tex}

\section{Simulation Results \label{sec:simulation}}


\subsection{Illustrative Example}
\noindent 
Figure~\ref{fig:sim_example} shows snapshots from a six vs. five scenario (one defender is out of the frame) with $\Qest=0$ and $\lgsi_k\leq 1,\forall\,k$.
\begin{figure}[t]
\centering
\includegraphics[width=.48\textwidth]
{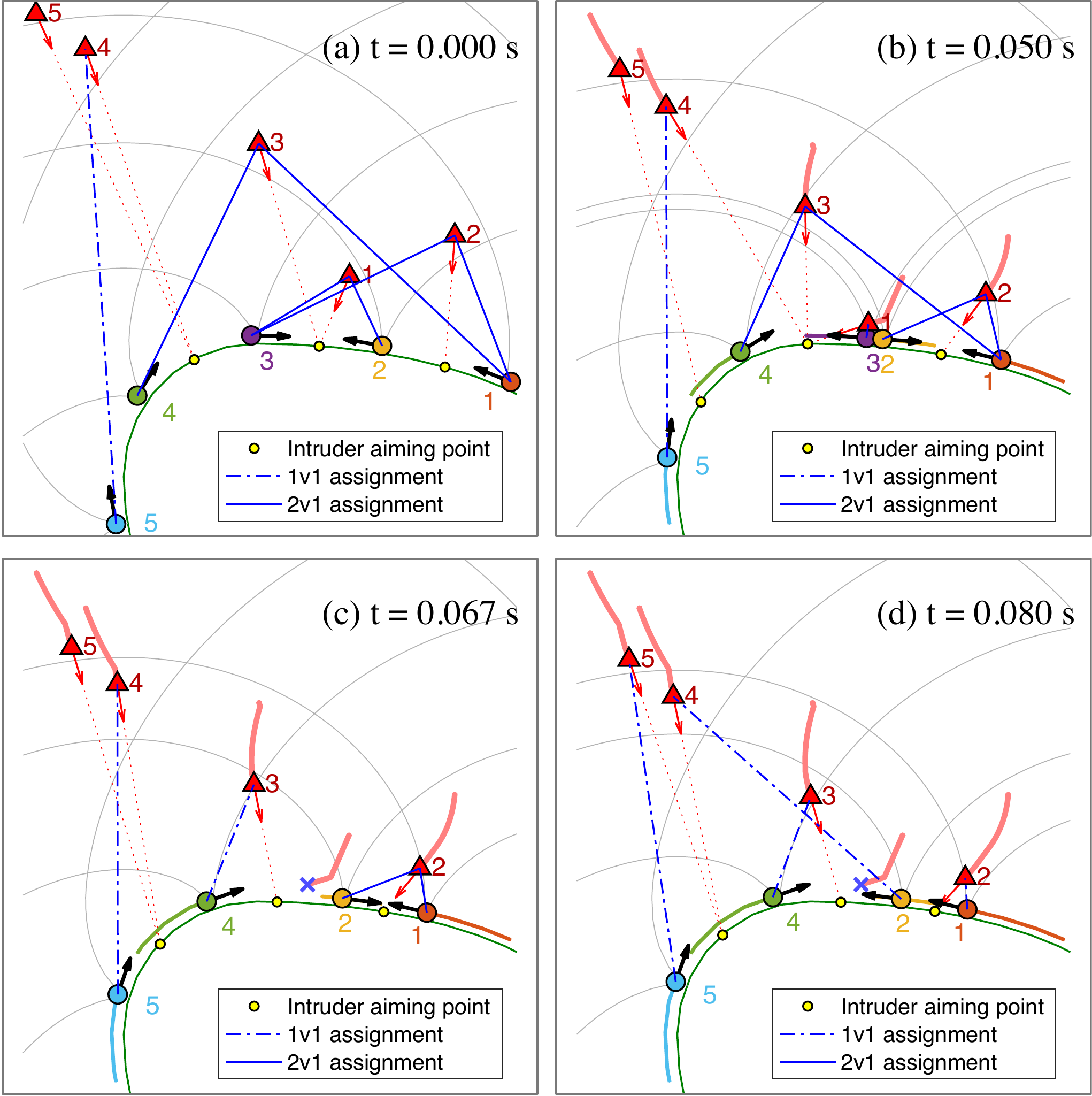}
\caption{Simulation snapshots.
}
\label{fig:sim_example}
\vspace{-10pt}
\end{figure}
Since there is no region with $q_k>0$, no team is formed among the intruders, and each plays an individually-optimal behavior against all the defender robots.
Each snapshot highlights the moment when the \sequential\ reduces into a smaller structure.
From $t=0.050$~s to $0.067$~s, the transition is caused by $A_3$ moving into $\dwinr(D_4)$.
The pair $(D_1,D_2)$ has an implicit assignment to $A_5$.
At $t=80$, the \implicit\ structure is resolved, and now all the intruders have 1v1 assignments that guarantee perfect defense.

Note that the MM policy will treat $A_1$ as uncapturable and assign 1v1 defenses so that the other four defenders will be captured, resulting to $\Qmm=1$.
The only way, in this scenario, to guarantee capture of both $A_1$ and $A_2$ is for $D_2$ to perform two 2v1 defenses sequentially (with $D_3$ and then with $D_1$), and the LGR defense algorithm achieves this without any explicit look-ahead planning.

\subsection{Statistical results}
\noindent
This section presents the difference between the expected performance of MM-defense and LGR-defense, measured by the upper bounds $\Qmm$ and $\Qest$.
Since the scores are greatly affected by the initial configuration, we compute the average score from randomly generated initial configurations.
To study the effect of initial distance of the intruders to the perimeter, we randomly select the azimuthal positions of the intruders for each distance.
We test two cases for the defenders' initial configuration: (i)~uniformly spaced, and (ii) randomly placed on the perimeter.

\begin{figure}[t]
\centering
\includegraphics[width=.48\textwidth]
{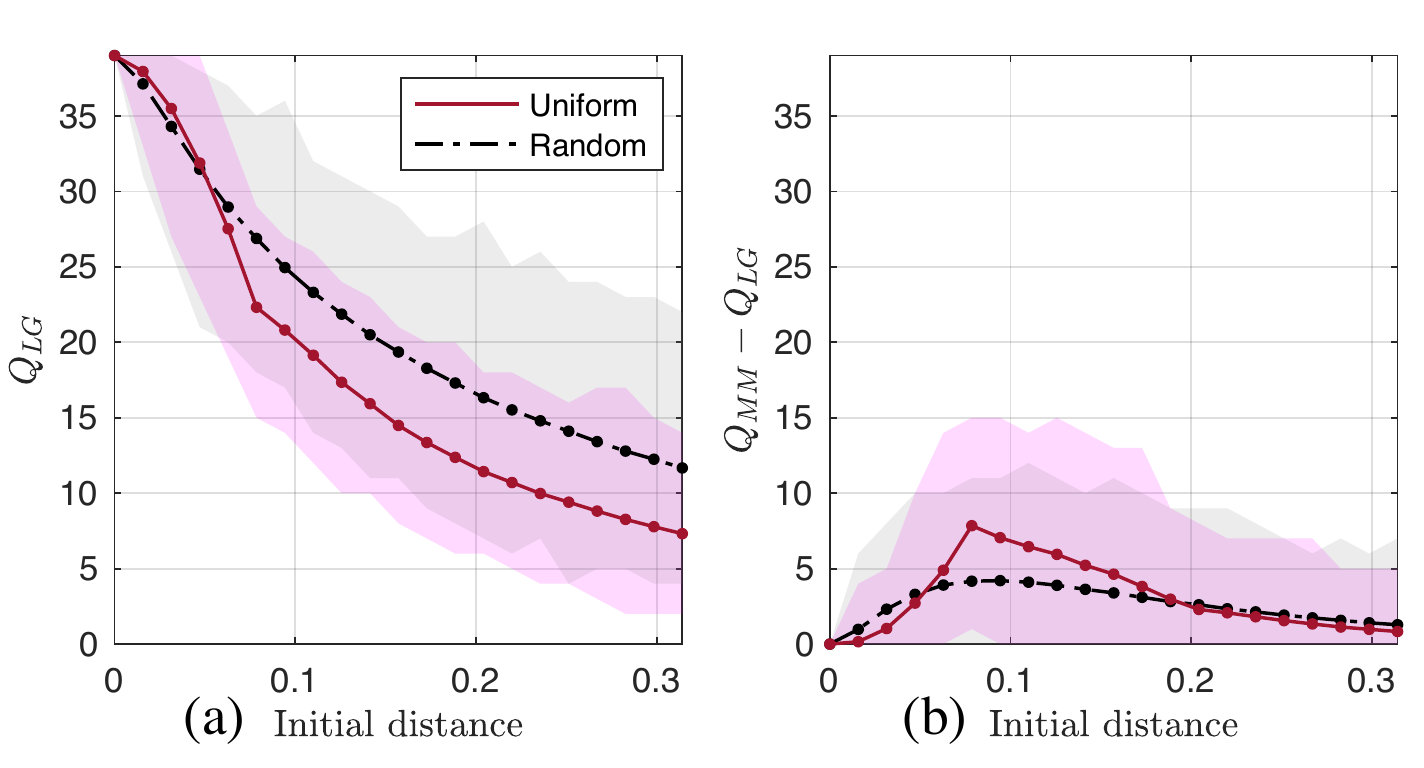}
\caption{
Score statistics with varying initial intruder distance, for uniformly (solid line) and randomly (dashed line) placed defenders.
(a) Score upper-bound $\Qest$.
The lines indicate the mean, and the shaded area is the envelope containing min and max.
(b) Suboptimality of the upper-bound $Q_{MM}$.
}
\label{fig:statistics}
\vspace{-10pt}
\end{figure}
Figure~\ref{fig:statistics}a shows how the scores change as a function of the initial distance.
The parameters are: $N_D=40$, $N_A=39$, $\nu=1$, and 10,000 initial configurations are tested for each distance.
The scores monotonically decrease as the intruders' starting distance increases and eventually converge to zero (not shown for clarity).
These values represent upper bounds on the number of intrusions which LGR policy can certify at the beginning of a game, and so might also be used to inform higher level coalition forming algorithms as in \cite{shishika2019cdc}.


Figure~\ref{fig:statistics}b shows the suboptimality of the bound $\Qmm$ provided by the MM policy. The difference satisfies $\Qmm-\Qest\geq 0$ because the LGR defense always performs at least as good as the MM policy.
The difference is small at short distance because almost all intruders can trivially score.
The difference increases with more opportunities for the cooperative 2v1 defenses, and it reduces again at larger distance
since the easier configuration allows the MM to perform near optimal: i.e., defenders do not need to employ cooperative strategy to capture intruders.

The shaded envelope containing the maximum difference highlights the benefit of using the LGR defense policy.
Consider, for example, the situation in which the intruders are detected at the distance $0.1$ units from the perimeter, and suppose they can select the initial azimuthal positions.
If the intruders select a formation possibly with the knowledge that the defenders employ MM strategy, the gap in the score can be as large as $\Qmm-\Qest\approx 15$, which is significant given the score $\Qest \approx 20 \pm 7$ at this distance. 

The statistical results on the computation time are  provided in \Appendix.




\section{Conclusion \label{sec:conclusion}}
\noindent
This paper presents a cooperative defense strategy for a variant of the multi-player reach-avoid game.
The proposed LGR defense policy runs in polynomial time and outperforms the existing assignment-based policy.
In fact, LGR defense performs optimally under certain conditions, where the performance is measured by the number of intruders that reaches the target.
The improved performance compared to the maximum-matching formulation is due to the incorporation of cooperative defense, in which a pair of defenders pursue an intruder from both sides.
The low computational complexity compared to a maximum-independent-set formulation is achieved by the proposed decomposition method into local games.

There are various aspects to be addressed in the future work, including the extension to three dimensions, environments with obstacles, distributed algorithm with relaxed assumptions on the information structure (e.g., local sensing and communication), and cooperative defense with heterogeneous teams.

%

\section*{ACKNOWLEDGMENT}
We gratefully acknowledge useful discussions with Chris Kroninger at ARL and Ken Hayashima.

\if\RALversion0
\input{appendix.tex}
\fi

\end{document}

%% file: extended_structures.tex
\section{Extended Engagement Types} \label{sec:extended_engagements}
\noindent 
We discuss two additional types of engagements that our defense strategy in Sec.~\ref{sec:defense_policy} leverages: the \emph{\implicit} and the \emph{\sequential}.
The distinct feature of these engagements is the aspect of dynamically changing assignments which the existing MM and MIS approaches fail to capture.

\subsection{Implicit Assignment Structure}
\begin{definition}
The \emph{\textbf{\implicit}} structure is a two vs.\ two engagement where the defender pair is assigned to one intruder in $\dwinrp(i,j)$, and the other intruder is in $\dwinr(i)\cap\dwinr(j)$, defined as the \implicit\ zone.
\end{definition}
\begin{figure}[t]
\centering
\includegraphics[width=.48\textwidth]
{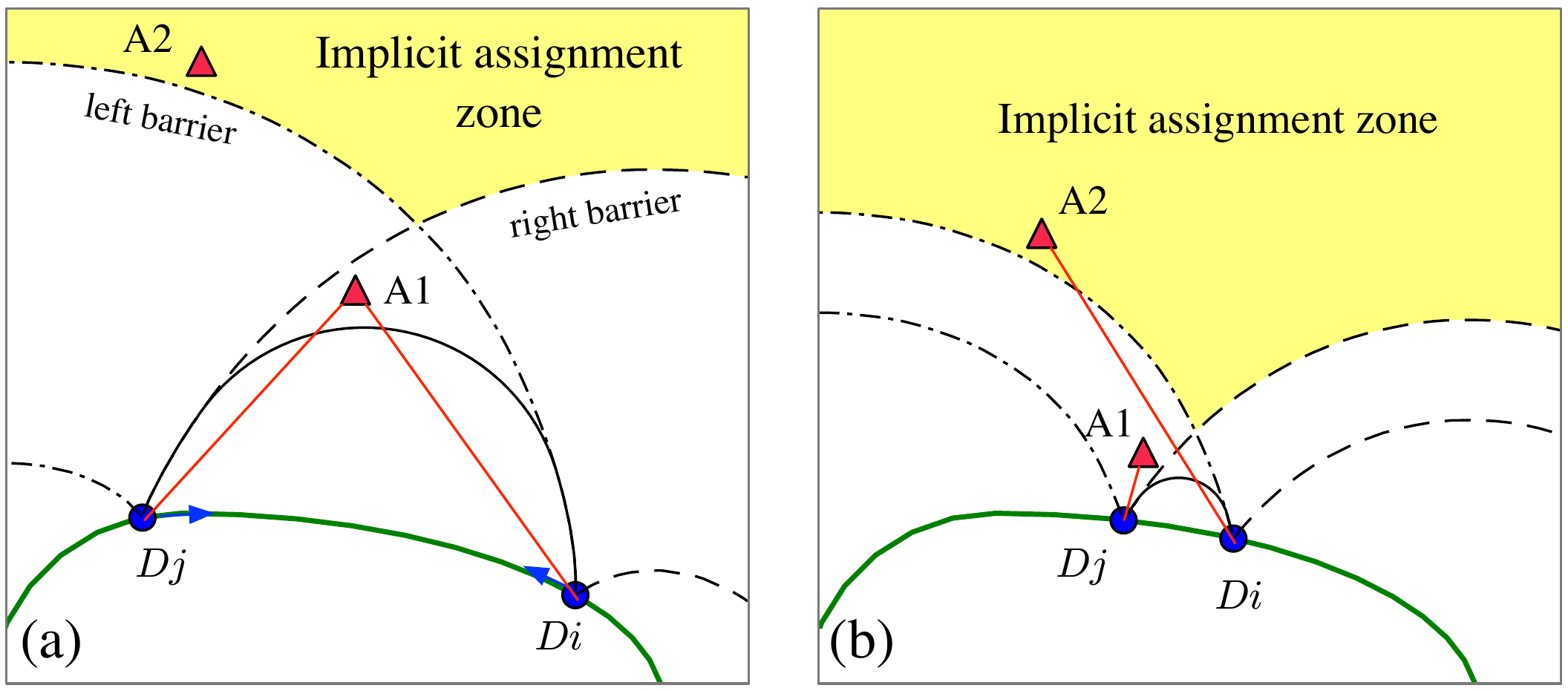}
\caption{The \implicit\ structure.
(a) Initial configuration where $A_1$ has 2v1 assignment, and $A_2$ is in the \implicit~zone.
(b) The structure is resolved into two 1v1 assignments.
}
\label{fig:unassigned_defense}
\vspace{-10pt}
\end{figure}
See Fig~\ref{fig:unassigned_defense}a for an example.
First note that MM assignment would ignore $A_1$ since it is not capturable by any individual defender.
The MIS assignment chooses from assigning (i) the pair to $A_1$, or (ii) one defender to $A_2$ (same as MM), but cannot have both because one of the defenders will have overlapping assignments.
Importantly, the two choices are equally good in the MIS analysis since they both guarantee one capture.
What MIS fails to see is that initially assigning the pair to $A_1$ is the optimal choice here, and in fact it leads to the capture of both intruders:

\begin{lemma}\label{lem:unassigned_defense1}
The \emph{\implicit} structure turns into two 1v1 assignments before any intruder reaches the perimeter.
\end{lemma}
\begin{proof}
From the 2v1 assignment against $A_1$, $\di$ moves ccw and $\dj$ moves cw.
Under this movement, $A_2$ remains in $\dwinr(i)\cap\dwinr(j)$ because it cannot penetrate $\di$'s left barrier nor $\dj$'s right barrier (Prop.~\ref{prop:1v1barrier}).
Before they meet at the midpoint, $A_1$ will move into either $\dwinr(i)$ or $\dwinr(j)$ (Prop.~\ref{prop:pincer}).
Suppose without the loss of generality $\mf x_{A_1}\in \dwinr(i)$, then $A_2$ will get 1v1 assignment to $\dj$. 
\end{proof}

We use the term ``implicit'' to highlight the fact that no defender is ``explicitly'' assigned to $A_2$.
Nevertheless, the greedy behavior against $A_1$ guarantees the capture of $A_2$ as well.
The second structure (i.e., a combination of configuration and assignment) is defined next.

\subsection{Sequential Paired Defense Structure}
\begin{definition}
The \emph{\textbf{\sequential}} structure is an $n+1$ vs.\ $n$ engagement generated by $n$ interrelated 2v1 games with overlapping defenders.
\end{definition}

\begin{figure}[t]
\centering
\includegraphics[width=.48\textwidth]
{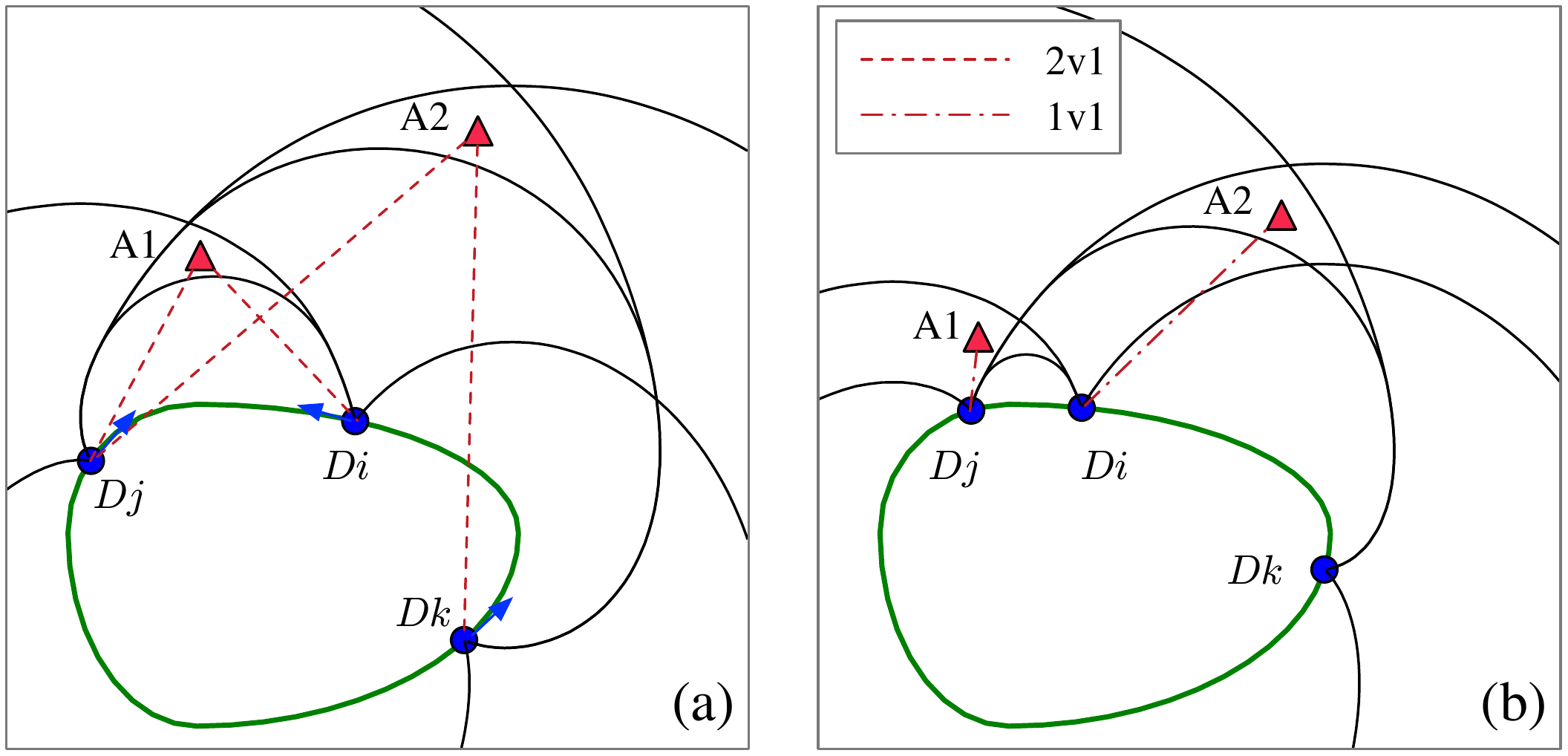}
\caption{
Sequential 2v1 structure with $n=2$.
(a) The defender $D_j$ has two 2v1 assignments, and cw motion is required for both.
(b) The structure is resolved into two 1v1 assignments.
}
\label{fig:sequential_2v1}
\end{figure}

Fig.~\ref{fig:sequential_2v1}a shows the simplest example of such configuration.
Note that we are now considering an assignment that is prohibited in the MIS formulation; defender $D_j$ has two 2v1 games assigned to it.
Also note that, by definition, each defender pair can have at most one intruder assigned to it (otherwise, it will not be $n+1$ vs. $n$).
The key concept to analyze the success in this configuration is defined next:

\begin{definition}
The overlapping 2v1 assignments are \emph{\textbf{conflicting}} if they require opposing (cw and ccw) direction of motion from the defender.
\label{def:conflict}
\end{definition}

For example, $D_j$ in Fig.~\ref{fig:sequential_2v1}a has non-conflicting assignments, because the two assignments both require cw motion.
Despite the overlap, it can behave optimally for both $A_1$ and $A_2$, which leads to the resolved configuration in Fig.~\ref{fig:sequential_2v1}b.
We formally show in Sec.~\ref{sec:perfect_defense} that the \sequential\ leads to the capture of all intruders if they are non-conflicting.

\subsection{Independent Local Game Score}
Before we present our defense algorithm, we extend the concept of LGS defined for LGRs in Lemma~\ref{lem:lgs} to I-LGRs.
The \emph{independent local game score} (I-LGS) is defined as
\begin{gather}
\lgsi_k \triangleq \hat{n}_A^k - n_D^k =  q_k + \Delta n_A^k,
\end{gather}
where $\Delta n_A^k$ is the number of intruders in the paired defense region $\dwinrp^k$.

The local game score of $q_k=0$ was a critical number in Sec.~\ref{sec:QLG} since it meant that any additional intruder in the LGR will immediately lead to a positive intruder score (Lemma~\ref{lem:lgs}). 
The I-LGS has another critical number: $\lgsi_k=1$.
Unlike $q_k=1$, this configuration does not immediately lead to intruder score; suppose $q_k \leq 0$ and $\lgsi_k = 1$, then one intruder in $\dwinrp^k$ can be captured by the pincer maneuver of the boundary defenders, and the rest may be captured by the defender subteam $\subteam_D^k$.
What $\lgsi_k=1$ immediately tells us is that the boundary defenders $(D_R^k,D_L^k)$ need to perform a pincer maneuver.
Otherwise the intruder in $\dwinrp^k$ can enter $\awinrc^k$ to achieve $q_k>0$, which guarantees a score.

%% file: defense_algorithm2.tex
\section{The LGR Defense Policy \label{sec:defense_policy}}
\noindent 
This section present the LGR-defense policy and its performance guarantees.
Combined with the intruder team strategy in Theorem~\ref{thm:Qlow}, we will discuss the optimality of the strategy.

\subsection{The LGR Defense Algorithm}
\noindent
The LGR defense strategy is presented in Algorithm~\ref{alg:LGRdefense}.
It is continuously run throughout the game to update the defender to intruder assignments that uniquely define each defender's direction of motion.
{As illustrated in Sec.~\ref{sec:extended_engagements}, the assignments will dynamically change as the game evolves in time.}
\begin{algorithm}[h]
\caption{LGR Defense \label{alg:LGRdefense}}
\begin{small}
\begin{algorithmic}[1]
\State Remove uncapturable intruders using Algorithm~\ref{alg:removal}
\For{every time step}
\State Assign 2v1 defense using Algorithm~\ref{alg:2v1assignment}
\State Assign 1v1 defense using Algorithm~\ref{alg:1v1assignment}
\State {\bf Return}: assignments (i.e., direction of motion)
\EndFor
\end{algorithmic}
\end{small}
\end{algorithm}

First, Algorithm~\ref{alg:removal} removes/ignores intruders from the game to consider a virtual game with $\Qest=0$. 
The purpose is to ignore the intruders that cannot be captured, and make sure all other intruders will be captured.
This algorithm needs to be run only once at the beginning of the game.
\begin{algorithm}[h]
\caption{Removal of uncapturable intruders \label{alg:removal}}
\begin{small}
\begin{algorithmic}[1]
\State Initialize: $\aignore=\emptyset$, and $i=1$ 
\While {$\Qest>0$}
    \State $Q_{new}\gets$ compute $\Qest$ without $\{A_i,\aignore \}$
    \If {$Q_{new} < \Qest$}
        \State Append $A_i$ to $\aignore$
        \State $\Qest\gets Q_{new}$
    \EndIf
    \State $i\gets i+1$
\EndWhile
\State {\bf Return}: $\aignore$
\end{algorithmic}
\end{small}
\end{algorithm}

Then the cooperative 2v1 defense and the individual 1v1 defense are considered sequentially.
In this paper, we restrict ourselves to the case when $\lgsi_k\leq 1,\,\forall\,k$ after the removal of uncapturable intruders.
The assignment of cooperative 2v1 games is presented in Algorithm~\ref{alg:2v1assignment}.
 \begin{algorithm}[h]
 \caption{Assign 2v1 Defense \label{alg:2v1assignment}}
 \begin{small}
 \begin{algorithmic}[1]
\State Initialize: $\mf A_\text{assgn}=\emptyset$, $\mf D_\text{2v1}=\emptyset$, and $\mf D_\text{1v1}=\dteam$
 \While{$\exists\,k$ s.t. $\lgsi_k = 1$}
\State $m \gets \argmax_k n_D^k$ for $k$ s.t. $\lgsi_k = 1$
     \State $A_j \gets$ select one intruder in $\dwinrp^m$
     \State Append $A_j$ to $\mf A_\text{assgn}$, and $m$ to $\mf D_\text{2v1}$
     \State Remove $D_R^m$ and $D_L^m$ from $\mf D_\text{1v1}$
     \State Recompute $\lgsi_k$ for all $k$ with $\ateam \setminus \mf A_\text{assgn}$
 \EndWhile
\State {\bf Return}: $\mf D_\text{2v1}$, $\mf D_\text{1v1}$ and $\mf A_\text{assgn}$
 \end{algorithmic}
 \end{small}
 \end{algorithm}

The I-LGRs with $\lgsi=1$ are visited sequentially from the largest to the smallest in terms of the subteam size.
In each iteration, one intruder in $\dwinrp$ is assigned to the boundary defenders (line 5) regardless of whether the defenders already have other 2v1 assignments or not, which potentially generates the \sequential\ structure.
The indices saved in $\dteam_\text{2v1}$ is sufficient to know which defender should perform a pincer movement with which pair.
The already assigned intruders are removed from the computation of $\lgsi$ (line 8), and the while loop terminates when all the I-LGRs have $\lgsi\leq 0$, which occurs in less than $N_A$ iterations.
Since the bottleneck in the while loop is the recalculation of $\lgsi$'s which is $O(N_D^2N_A)$, the time complexity of Algorithm~\ref{alg:2v1assignment} is $O(N_D^2N_A^2)$.

 \begin{algorithm}[h]
 \caption{Assign 1v1 Defense \label{alg:1v1assignment}}
 \begin{small}
 \begin{algorithmic}[1]
\State Inputs: $\ateam_\text{assgn}$, $\dteam_\text{2v1}$, and $\dteam_\text{1v1}$ 
 \State $\ateam_\text{1v1}\gets \ateam \setminus \ateam_\text{assgn}$
\State Initialize a bipartite graph with $\{\dteam_\text{1v1},\dteam_\text{2v1}\}$ and $\ateam_\text{1v1}$
 \For{$k$ in $\dteam_\text{2v1}$}
     \If {inequality in Eq.~\eqref{eq:ineligible} is satisfied}
     	\State Draw edges to $\forall\,A_j$ s.t. $\mf x_{A_j}\in\dwinr(D_R^k)\cap\dwinr(D_L^k)$
     \EndIf
 \EndFor
 \For{$D_i$ in $\dteam_\text{1v1}$}
     \State Draw edges to $\forall\,A_j$ s.t. $\mf x_{A_j}\in\dwinr(D_i)$
 \EndFor
 \State Solve maximum matching and save the edges for $\dteam_\text{1v1}$
 \State {\bf Return}: 1v1 assignments (edges for $\dteam_\text{1v1}$)
 \end{algorithmic}
 \end{small}
 \end{algorithm}

When assigning 1v1 games, Algorithm~\ref{alg:1v1assignment} accounts for the \implicit\ structure.
The implicit assignment is still applicable to the pairs in \sequential\ structure, but with a restriction.
Noting that each structure involves $n+1$ defenders and $n$ intruders, we can only consider one additional ``implicit'' assignment, otherwise there will be more intruders than the defenders.
To guarantee capture, it is necessary and sufficient for the additional intruder to be in the \implicit\ zone of the outermost defender pair (e.g., $(D_k,D_j)$ in Fig.~\ref{fig:sequential_2v1}).
In summary, a 2v1 defender pair, $\dboundary^k$, can have an \implicit\ if $D_R^k$ or $D_L^k$ are not used in any other 2v1 defense that is associated to a larger I-LGR, i.e.,
\bnl
n_D^k > n_D^m, \forall m\in \dteam_\text{2v1}\;\text{s.t.}, \dboundary^m\cap\dboundary^k\neq\emptyset
\label{eq:ineligible}
\enl

The solution of the matching problem (line 9) only affects the behavior of the defenders in the set $\dteam_\text{1v1}$, but not the pairs in $\dteam_\text{2v1}$ because the assignments are only ``implicit'' for them.
This is why we only need the edges for $\dteam_\text{1v1}$.
The computational bottleneck in Algorithm~\ref{alg:1v1assignment} is the maximum matching, which is known to have efficient polynomial time algorithms \cite{Chen2017}.
Noting that $|\{\dteam_\text{1v1},\dteam_\text{2v1}\}|\leq N_D$ and $|\ateam_\text{1v1}|\leq N_A$, the number of nodes in the bipartite graph is at most $N_D+N_A$, which leads to $O((N_A+N_D)^w)$ time with $w<2.5$ \cite{Mucha2004}.

An example of the assignments given by Algorithm~\ref{alg:LGRdefense} is shown in Figure~\ref{fig:sim_example}a simulation results.
Notice that a defender can have multiple 2v1 games assigned to it.
Algorithm~\ref{alg:2v1assignment} generated the \sequential\ structure involving $D_1,D_2,D_3$ and $D_4$.
After the 2v1 assignments, $A_4$ and $A_5$ are left for 1v1 assignments.
The pair $(D_1,D_4)$ is eligible for the \implicit, and only $A_5$ is in the \implicit\ zone.
Therefore, the maximum matching in Algorithm~\ref{alg:1v1assignment} selects $A_4$ to be pursued by $D_5$ and leaves $A_5$ to be implicitly assigned to $(D_1,D_4)$.
We will prove the performance guarantees in the remaining sections.

\subsection{Performance Guarantees} \label{sec:perfect_defense}
\noindent
This section proves that the LGR defense in Algorithm~\ref{alg:LGRdefense} guarantees the overall score to be zero if the game starts in a configuration with $\Qest(\mf z_0)=0$ and $\lgsi_k(\mf z_0)\leq~1,\forall\,k$.
We first provide instantaneous properties:

\begin{lemma}\label{lem:completeness}
If $\Qest(\mf z)=0$ and $\lgsi_k(\mf z)\leq 1,\forall\,k$, then no intruder is unmarked: i.e.,
every intruder has either (i) 1v1 assignment, (ii) 2v1 assignment, or (iii) implicit assignment.
\end{lemma}
\begin{proof}
It suffices to show that, after Algorithm~\ref{alg:2v1assignment}, the intruders without 2v1 assignment gets either 1v1 or implicit assignment.
We use Hall's marriage theorem \cite{Hall} to show that the maximum matching in Algorithm~\ref{alg:1v1assignment} covers all the intruders.
For any subset $W\subseteq \ateam_\text{1v1}$, let $\mc{N}(W)$ denote the neighbor set: i.e., the subset of defenders that can capture at least one in $W$.
Hall's theorem states that all intruders in $\ateam_\text{1v1}$ are covered if $|W|\leq |\mc{N}(W)|$ for all $W\subseteq \ateam_\text{1v1}$.

For $|W|=1$, it is easy to see by contradiction that every intruder has at least one edge:
if an intruder has no edge, then the smallest I-LGR that contains it will have $\lgsi \geq 1$,
which contradicts the termination of Algorithm~\ref{alg:2v1assignment}.

Now, suppose the inequality is true for $|W|=1,..,n$.
We will prove that the inequality is also true for $|W|=n+1$:
If the bipartite graph can be separated into $m$ connected components involving intruder subsets $w_1,.. w_m$, then we have $|\mc{N}(W)|=\sum_i^m |\mc{N}(w_i)|$ because the defender nodes for each subset $w_i$ are disjoint.
Since we have $|\mc{N}(w_i)|\geq |w_i|$ from the supposition that the inequality is true for $|w_i|\leq n$, we obtain $|\mc{N}(W)| \geq \sum_i^m |w_i| = |W|$.
If the bipartite graph is connected and $|\mc{N}(W)|\leq n$, there exists an I-LGR whose defender subteam is $\mc{N}(W)$ and the intruder subteam is $W$.
Then this region has the score $\lgsi\geq(n+1)-n=1$.
Since this contradicts the termination of Algorithm~\ref{alg:2v1assignment}, we obtain $|\mc{N}(W)|\geq n+1=|W|$.

With mathematical induction, we have shown that $|W|\leq |\mc{N}(W)|$ for $W\subseteq \ateam_\text{1v1}$ with any size.
Therefore, maximum matching in Algorithm~\ref{alg:1v1assignment} covers all the intruders.
\end{proof}

\begin{lemma}\label{lem:sequential_2v1noconflict}
If $\Qest(\mf z)=0$ and $\lgsi_k(\mf z)\leq1,\,\forall\,k$, then the \sequential\ structure generated by Algorithm~\ref{alg:2v1assignment} has no conflict (see Definition~\ref{def:conflict}).
\end{lemma}
\begin{proof}
%
For $\di$ to have conflicting 2v1 assignments, the following condition is necessary: $\di$ has two I-LGRs, $l$ and $r$ with $\lgsi_l=\lgsi_r=1$, where $D_R^l = \di$ and $D_L^r=\di$.
If $D_i$ gets both assignments, then there is a conflict.

Consider $\awinri^m\triangleq\awinri(D_R^r,D_L^l)$, which contains both $\awinri^l$ and $\awinri^r$.
The score satisfies $\lgsi_m\geq (n_D^l+1)+(n_D^r+1)-(n_D^l+n_D^r+1)=1$.
We also know $\Qest=0 \Rightarrow q_m=0$, implying that $\dwinrp^m$ has at least one intruder to be assigned to $\dboundary^m$.
Since this assignment reduces one intruder from either $l$ or $r$, $\di$ only needs to perform 2v1 defense for one of the two: i.e., $\di$ does not get both assignments.
\end{proof}

The above results are used to study how the assignments and scores evolve over time:

\begin{lemma}\label{lem:sequential_2v1resolve}
If initially $\Qest(\mf z_0)=0$ and $\lgsi_k(\mf z_0)\leq1,\,\forall\,k$, then the number of 2v1 assignments reduces over time, and the conditions $\Qest(\mf z)=0$ and $\lgsi_k(\mf z)\leq1,\,\forall\,k$ are preserved.
\end{lemma}
\begin{proof}
By Lemma~\ref{lem:sequential_2v1noconflict}, each defender pair performs optimally against the assigned intruder.
From Prop.~\ref{prop:pincer}, the intruder currently in $\dwinrp$ moves into the defender winning region of one of the boundary defenders.
Among all the 2v1 games, consider the one that achieves this transition first.
During the transition, the intruder leaves an I-LGR that it was contained in (reducing its I-LGS from 1 to 0), but it does not enter any new I-LGR.
Such transition guarantees that the scores $q_k$ and $\lgsi_k$ are both non-increasing for all $k$ (see \Appendix\ for more details).
We now have one less 2v1 games.
Since the conditions for Lemma~\ref{lem:sequential_2v1noconflict} are still true, similar transitions will repeat until there is no more 2v1 assignments.
\end{proof}

\begin{theorem}\label{thm:perfect_defence}
If initially $\Qest(\mf z_0)=0$ and $\lgsi_k(\mf z_0)\leq 1, \forall\,k$, then the LGR defense strategy (Algorithm~\ref{alg:LGRdefense}) guarantees the overall score to be zero.
\end{theorem}
\begin{proof}
Lemmas \ref{lem:unassigned_defense1}~and~\ref{lem:sequential_2v1resolve} show that the number of \implicit s and 2v1 assignments reduces as the game evolves.
Since the conditions $\Qest=0$ and $\lgsi\leq1$ are preserved throughout the game (Lemma~\ref{lem:sequential_2v1resolve}), we can invoke Lemma~\ref{lem:completeness} to conclude that every intruder gets 1v1 assignment before it reaches the perimeter.
Therefore, no intruder can score.
%
\end{proof}

\subsection{Saddle-point equilibrium}
\label{sec:saddle_point}

\noindent
This section addresses a general case when the intruder team can guarantee non-zero score: $\Qest>0$.
First, the following lemma shows that exactly $\Qest$ intruders need to be removed in Algorithm~\ref{alg:removal}.
\begin{lemma}\label{lem:removal}
If $\Qest(\mf z_0)>0$, then there exists a set of $\Qest(\mf z_0)$ intruders whose removal generates a new game with $\Qest=~0$.
\end{lemma}
\begin{proof}
When $\Qest(\mf z)>0$, one can show by contradiction that there exists at least one intruder whose removal reduces $\Qest$ by $1$ (see \Appendix\ for more details).
After removing this intruder, the same argument holds if $\Qest$ is still positive.
Repeating this process, the score reduces to zero after removing $\Qest(\mf z)$ intruders.
\end{proof}

 Following the proof of Lemma~\ref{lem:removal}, we can find the set $\aignore$ in Algorithm~\ref{alg:removal} by visiting every intruder and asking whether it reduces $\Qest$ or not (Algorithm~\ref{alg:removal}).
Importantly, the answer remains the same whether $\Qest$ is already reduced or not.
Therefore, each intruder needs to be visited only once; Algorithm~\ref{alg:removal} terminates within less than $N_A$ iterations.
Since the bottleneck in the while loop is the computation of $\Qest$ (line 3), the time complexity is $O(N_AN_D^4)$.

The main result is summarized in the following:
\begin{theorem}\label{thm:Qup}
If initially $\hat{q}_k(\mf z_0)\leq 1,\,\forall\,k$ after removing intruders found by Algorithm~\ref{alg:removal}, then the LGR defense policy $\controlset_D^*$ guarantees
\bql
Q(\mf z_0; \controlset_D^*,\controlset_A)\leq  \Qest(\mf z_0), \; \forall \, \controlset_A\in U_A,
\label{eq:uDguarantee}
\eql
where $U_A$ is the set of all permissible intruder strategies.
\end{theorem}
\begin{proof}
From Lemma~\ref{lem:removal}, ignoring $\Qest(\mf z_0)$ intruders  generates a virtual game with $\Qest=0$.
All intruders in this virtual game are captured since the conditions for Theorem~\ref{thm:perfect_defence} are satisfied.
In the original game, at most $\Qest(\mf z_0)$ intruders that are ignored by the defenders can possibly score.
\end{proof}

Finally, the combination of the strategy $\mf \controlset_A^*$ from Theorem~\ref{thm:Qlow} and $\controlset_D^*$ from Theorem~\ref{thm:Qup} gives us the following result:
\begin{corollary}
If $\hat{q}_k(\mf z_0)\leq 1,\,\forall\,k$ after removing $\Qest(\mf z_0)$ intruders found by Algorithm~\ref{alg:removal}, then
\bql
Q(\controlset_D^*,\controlset_A)\leq  Q(\controlset_D^*,\controlset_A^*) \leq Q(\controlset_D,\controlset_A^*),
\eql
i.e., the strategy set $(\controlset_D^*,\controlset_A^*)$ gives a saddle-point equilibrium, and the value of the game is $Q^*(\mf z_0) = \Qest(\mf z_0)$.
\end{corollary}
If the intruders stick to $\controlset_A^*$, then the defender team cannot reduce the score by changing their strategy.
Similarly, if the defenders keep $\controlset_D^*$, then the intruders cannot increase the score by changing their strategy.

%% file: appendix.tex
\section*{APPENDIX}
\subsection{Proof of Lemma~\ref{lem:sequential_2v1resolve}}\label{sec:proof_2v1resolve}
From Lemma~\ref{lem:sequential_2v1noconflict}, there are no conflicting 2v1 assignments.
Therefore, looking at each 2v1 game separately, the defender pair performs optimally against the assigned intruder.
Recalling the 2v1 analysis, the intruder in $\dwinrp$ will eventually move into the defender winning region of either of the two boundary defenders (Prop.~\ref{prop:pincer}).
Consider the first transition that occurs among all 2v1 games,\footnote{The case in which multiple transitions occur simultaneously can also be treated similarly.} and let $k_1$ be the index of the associated I-LGR.
Clearly, $\lgsi_{k_1}$ changes from 1 to 0 with this transition.

Let us examine how $\lgsi$ changes for other I-LGRs.
Observe that the transition occurred by the intruder leaving $\awinri^{k_1}$ (see Fig.~\ref{fig:2v1barrier}).
Also observe that any I-LGR that contains this intruder immediately after the transition had already contained it before the transition.
From the above two observations, it is clear that $\lgsi_k$ is non-increasing for all $k$.

Let us also examine if $q_k$ can possibly become positive for any $k$ with the transition.
For any $q_k=0$ to increase to a positive number, either of the following two needs to occur: (i) an intruder inside $\awinri^k$ enters $\awinrc^k$, or (ii) an intruder outside $\awinri^k$ enters $\awinrc^k$.
The first case is equivalent to an intruder moving from $\dwinrp^k$ to $\awinrc^k$, where $\lgsi_k = 1$.
This cannot occur due to Proposition~2.
The second case entails an intruder to enter $\awinri^k$ because of the relation $\awinrc^k\subset\awinri^k$.
However, this cannot occur either from the discussion in the previous paragraph.
Therefore, $q_k\leq 0,\forall\, k\Leftrightarrow \Qest=0$ will be preserved with the transition.

We have shown that there are one less I-LGR with $\lgsi=1$, implying that there are one less 2v1 assignments, and the condition $\Qest=0$ is still true after the transition.
Since the conditions for no conflict (Lemma~\ref{lem:sequential_2v1noconflict}) are still true, the next transition that reduces the number of 2v1 games eventually occurs.
By repeating this process, we see that the number of 2v1 assignments monotonically decreases as the game progresses.

\subsection{Proof of Lemma~\ref{lem:removal}}\label{sec:proof_removal}
We show by contradiction that there exists at least one intruder whose removal reduces $\Qest$ by one, when $\Qest>0$.
The contradiction will be derived by supposing that there is no such intruder, and showing that it leads to the existence of $\djset$ with $\Qint(\djset)>\Qest$, which contradicts the optimality of $\Qest$ in \eqref{eq:Qest}.

Suppose there is no intruder whose removal reduces $\Qest$.
Then there exists at least two sets $\djset_1^*$ and $\djset_2^*$ that gives the optimal score $\Qint(\djset_1^*) = \Qint(\djset_2^*) = \Qest$, with the constraint $\subteam_A^k\cap\subteam_A^l=\emptyset$ for $k\in\djset_1^*$ and $l\in\djset_2^*$, i.e., no common intruders.
The two sets do not share any intruder because if they did, then the removal of the common intruder will reduce both $\Qint(\djset_1^*)$ and $\Qint(\djset_2^*)$.

Now, take $k\in\djset_1^*$ with $q_k>0$ and consider merging it to $\djset_2^*$.
If $\awinrc^k\cap\awinrc^l=\emptyset$ for $\forall\, l\in\djset_2^*$, then $\djset_3 = \{k,\djset_2^*\}$ is a valid disjoint set and it has $\Qint(\djset_3)>\Qint(\djset_2^*)$, which contradicts  the optimality of $\djset_2^*$.
If $k$ intersects with $l\in\djset_2^*$, then let $m$ be the smallest LGR that contains both $k$ and $l$.
Since the regions intersect, we have $n_D^m<n_D^k+n_D^l$.
However, since they share no intruders, we have $n_A^m\geq n_A^k+n_A^l$.
Therefore, $q_m= n_A^m-n_D^m > n_A^k+n_A^l- n_D^k-n_D^l = q_k+q_l$, which again contradicts the optimality of $\djset_2^*$.
The case when $k$ intersects with multiple LGRs in $\djset_2^*$ can be treated similarly.

After removing one intruder, the same argument still holds if $\Qest$ is positive.
Repeating this process, the score $\Qest$ reduces to zero after removing $n$ intruders.
Also, observe that the order of the removal does not affect the result.

\subsection{Computation Time} \label{sec:computation_time}
Figure~\ref{fig:statistics_comp} shows how the scores and the computation time change with the number of intruders, $N_A$.
The number of defenders are $N_D=N_A+1$.
For each $N_A$, we tested 100 configurations.
The configurations are generated similarly to Sec.~\ref{sec:simulation}, where the initial distance is set to be one half of the defender spacing: i.e., $r_A= \frac{\pi}{N_D}$.

\begin{figure}[h]
\centering
\includegraphics[width=.48\textwidth]
{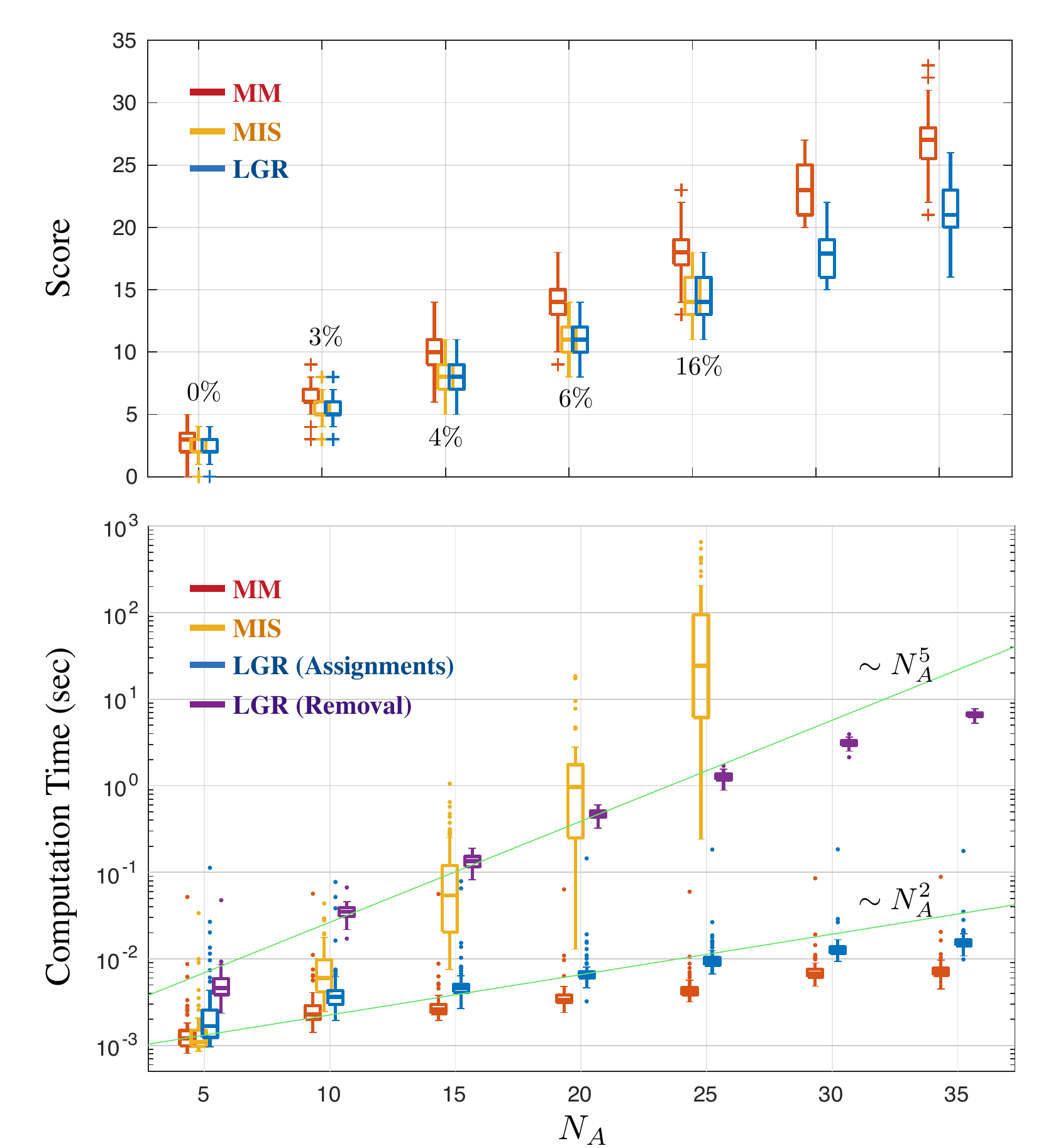}
\caption{Top figure shows the score bounds provided by three different policies.
The accompanied percentages indicate the fraction of trials in which $Q_\text{MIS}>Q_\text{LG}$.
The bottom figure shows the computation time in $\log$ scale.
The trend lines that scale with $N_A^2$ and $N_A^5$ are shown in green.
}
\label{fig:statistics_comp}
\end{figure}

The top figure shows that the MM policy suffers from suboptimal capture guarantees.
While the MIS policy frequently provides the optimal score $\Qest$, it still cannot always provide this bound as shown by the fraction of trials in which the score was $Q_\text{MIS}>Q_\text{LG}$. 

The bottom figure shows the computation time for the algorithms implemented in MATLAB, run on a laptop with a Core i7-7820HQ processor with 16 GB of memory.

The LGR-defense policy is considered into two parts: the removal of uncapturable intruders (Algorithm~2) which needs to be run only once at the beginning of the game, and the assignments (Algorithms~3 and 4) that needs to be run continuously throughout the game.
In Sec.~\ref{sec:saddle_point}, we showed that Algorithm~\ref{alg:removal} has the complexity $O(N_AN_D^4)$, which is equivalent to $O(N_A^5)$ in this experiment.
The comparison of the purple data points ``LGR (Removal)'' with the trend line $\sim N_A^5$ validates the above theoretical bound.

We showed in Sec.~\ref{sec:defense_policy} that Algorithms~\ref{alg:2v1assignment}
and \ref{alg:1v1assignment} have complexities $O(N_D^2N_A^2)$ and $O((N_A+N_D)^{2.5})$, equivalent to $O(N_A^4)$ and $O(N_A^{2.5})$ in this example.
The comparison of the blue data points ``LGR (Assignments)'' with the trend line $\sim N_A^2$ shows that the empirical growth in computation time is sub-quadratic.

The MM policy that has the bound $O(N_A^{2.5})$ also shows sub-quadratic complexity.
The MIS policy performs exhaustive search for the maximum independent set.
It is clear that the computation time is super-polynomial.
The result is shown only up to $N_A=25$, since for $N_A\geq 30$ there were examples where the search could not be completed.

This statistical result supports our claim that the LGR defense policy outperforms MM and MIS policy while maintaining scalability (i.e., polynomial time algorithm). 